\documentclass[letterpaper, 10 pt, conference]{ieeeconf}  

\IEEEoverridecommandlockouts                              

\usepackage{amsthm}
\usepackage{amsmath} 
\usepackage{amssymb}  
\usepackage{color}
\usepackage{empheq} 
\usepackage{epsfig} 
\usepackage{graphicx} 
\usepackage{subcaption}
\usepackage{times} 
\usepackage{wrapfig}
\usepackage{multicol}
\usepackage{caption}
\usepackage{cite}
\usepackage{relsize}

\usepackage{array}
\usepackage[]{algorithm2e}
\usepackage{algpseudocode}

\DeclareMathOperator{\trace}{Tr}
\renewcommand{\vec}[1]{\mathbf{#1}}
\newcolumntype{M}{>{\centering\arraybackslash}m{\dimexpr.25\linewidth-2\tabcolsep}}

\newtheorem{theorem}{Theorem}
\newtheorem{proposition}{Proposition}

\newtheorem{lemma}{Lemma}

\graphicspath{{Media/}}

\usepackage[capitalise]{cleveref}

\title{\LARGE \bf
Adaptive Communication Networks with Privacy Guarantees
}

\author{Atiye Alaeddini$^{1}$, Kristi Morgansen$^{2}$ and Mehran Mesbahi$^{3}$
\thanks{$^{1}$Atiye Alaeddini is with Institute for Disease Modeling, 3150 139th Ave SE, Bellevue, WA, 98005
        {\tt\small aalaeddini@idmod.org}}%
\thanks{$^{2}$Kristi Morgansen is with William E. Boeing Department of Aeronautics and Astronautics, University of Washington, Seattle, WA, 98195-2400
        {\tt\small morgansen@aa.washington.edu}}%
\thanks{$^{3}$Mehran Mesbahi is with William E. Boeing Department of Aeronautics and Astronautics, University of Washington, Seattle, WA, 98195-2400
        {\tt\small mesbahi@aa.washington.edu}. The research of M. Mesbahi has been supported by ARO grant W911NF-13-1-0340 and ONR grant N00014-12-1-1002.}%
}

\begin{document}

\maketitle
\thispagestyle{empty}
\pagestyle{empty}

\begin{abstract}

Utilizing the concept of observability, in conjunction with tools from graph theory and optimization, this paper develops an algorithm for network synthesis with privacy guarantees. In particular, we propose an algorithm for the selection of optimal weights for the communication graph in order to maximize the privacy of nodes in the network, from a control theoretic perspective. In this direction, we propose an observability-based design of the communication topology that improves the privacy of the network in presence of an intruder. The resulting adaptive network responds to the intrusion by changing the topology of the network-in an online manner- in order to reduce the information exposed to the intruder. 

\end{abstract}

\section{INTRODUCTION}

Networked dynamic systems consist of multiple dynamic units that are interconnected via a network. In recent years, the area of networked systems has received extensive attention from the research community. There are many examples of networked systems in our everyday lives such as social and transportation networks. Analysis of certain classes of complex networks such as biological networks, power grids, and robotic networks is of increasing interest in a number of scientific and engineering communities.

A network abstracts the communication topology for exchanging information between agents or nodes
in order to coordinate reaching a network-level goal. When different agents exchange sensitive data, one of the main concerns is ensuring privacy. 
In a dynamic setting, network observability captures the information content available to an individual agent in the network. In this work, we propose an adaptation mechanism for the network topology that aims to minimize observability, that in turn can be used to minimize the information exposed to an intruder. 

Observability-based design of a multi-agent network can be found in recent works~\cite{pequito2014optimal, kibangou2014observability, kia2015dynamic, alaeddini2016optimal}. The notion of observability has also been used in multi-robot localization~\cite{mariottini2005vision, zhou2008robot, huang2011observability, sharma2012graph}, social networks~\cite{golovin2011adaptive}, electric power grid management~\cite{wu1988real}, and biological systems~\cite{kang2009computational, chis2011structural, whalen2012observability}. Protocols with privacy guarantees have been previously addressed in~\cite{kefayati2007secure, huang2012differentially, manitara2013privacy, giraldo2014delay}, where the injection of random offsets into the states of the agents have been considered. Among other works that have examined network privacy, the authors of~\cite{pequito2014design} have considered the connection between privacy in the network and its observability. Privacy guarantees in~\cite{pequito2014design} ensure that each agent is unable to retrieve the initial states of non-neighboring agents in the network. Since all agents are potentially malicious, the optimal solution presented in ~\cite{pequito2014design} is inherently conservative, resulting in removing some edges to generate a topology that as many nodes as possible are part of the unobservable subspace of the graph.


The main contribution of this work is presenting the problem of privacy maximization in a network as a regret minimization problem. The present work considers the design of the weights on the edges of the network in an online manner in order to achieve maximum privacy. This is done using the regret minimization framework- in particular, we do not require any noise injection or edge removal as other works in network privacy have relied upon in the past. In this work, we assume that the structure of the network has been given-and that in fact-it is  connected. We then proceed to design the weights in the network in order to minimize the amount of information about a node to other agents in the network from the observability perspective. In order to guarantee privacy, our goal is proposing a network adaptation mechanism such that if a node is compromised by an intruder, the amount information leaked by this node to the rest of the network is minimized. 

\section{BACKGROUND}
\label{sec:prelim}

The interactions amongst agents in a network can be abstracted as an undirected graph $\mathcal{G}=(\mathcal{V},\mathcal{E},\vec{w})$. Each agent in the network is denoted as a node, and edges represent communication links between agents. The number of nodes will be denoted by $N$ and the number of edges as $M$. The node set, $\mathcal{V}$, consists of all nodes in the network. On the other hand, the edge set, $\mathcal{E}$, is comprised of pairs of nodes $\{i,j\}$, if nodes $i$ and $j$ are adjacent. Each edge is assumed to have a weight $w_i \in \mathbb{R}_{>0}$. The neighborhood set $\mathcal{N}(i)$ of node $i$ is composed of all agents in $\mathcal{V}$ that are adjacent to $v_i$. 
The weighting vector $\vec{w} \in \mathbb{R}_{>0}^M$ with dimension $M$, represents the weight on the edges. These edges are encoded through the index mapping $\sigma$ such that $l=\sigma(i,j)$, if and only if edge $l$ connects nodes $i$ and $j$. The edge weights will be denoted as $w_{ij}$ and $w_l$, interchangeably. The adjacency matrix is an $N \times N$ symmetric matrix with $\mathcal{A}[i,j] = 1$ when $\{i,j\} \in \mathcal{E}$, and $\mathcal{A}[i,j] = 0$, otherwise. The weighted adjacency matrix is an $N \times N$ symmetric matrix with $\mathcal{A}[i,j] = w_l$ when $\{i,j\} \in \mathcal{E}$ and edge $l$ connects them, and $\mathcal{A}[i,j] = 0,$ otherwise. The weighted degree $\delta_i$ of node $i$ is the sum of the weights of the edges connecting node $i$ to its neighbors. The weighted degree matrix, $\Delta_w$, is a diagonal matrix with $\delta_i$ as its $i$th diagonal entry. The incidence matrix $E(\mathcal{G})$ is an $N \times M$ matrix. Column $\sigma(i,j)$ of the incidence matrix is $\vec{e}_{i}-\vec{e}_{j}$, which is represented by $\vec{e}_{ij}$ in this paper. The graph Laplacian is defined as $L=\Delta_w-\mathcal{A}_w$, and is a positive semi-definite matrix, which satisfies $\vec{1}^TL =\vec{0}$ and $L \vec{1}=\vec{0}$. Here, $\vec{1}$ denotes a vector with all elements equal to one. 

If the matrix $A$ is positive semidefinite (respectively, positive definite), denoted as $A \succeq 0$ (respectively, $A \succ 0$), then all eigenvalues of $A$ are nonnegative (respectively, positive). We consider the following operation over matrices: the Hadamard product is a binary operation that takes two matrices of the same dimension, and produces another matrix where each element $[i,j]$ is the product of elements $[i,j]$ of the original two matrices:
$$(A \circ B)[i,j] = A[i,j] \cdot B[i,j]\,.$$
The Hadamard product is associative and distributive, and unlike the matrix product it is also commutative. It is known that the Hadamard product of two positive semi-definite matrices is positive semi-definite.

\subsection{Consensus Algorithm}
The consensus algorithm is a control policy, based on local information, that enables the agents to reach agreement on certain states of interest. 
%
Let us provide a brief overview of the consensus protocol.
Consider $x_i(t) \in \mathbb{R}$ to be the $i$th node's state at time $t$. The continuous-time consensus protocol is then defined as 
$$\displaystyle \dot{x}_i(t) = \sum_{\{i,j\} \in \mathcal{E}} w_{ij} \left(x_j(t)-x_i(t)\right)\,.$$
In a compact form with $\vec{x}(t) \in \mathbb{R}^N$, the corresponding collective dynamics is represented as $\dot{\vec{x}} = -L(\vec{w}) \vec{x}$. Thus, each agent only requires its relative state with respect to its neighbors for the consensus dynamics. The dynamics of a connected network performing the consensus algorithm converges to agreement on the state \cite{olfati2007consensus}.

\subsection{Observability of a Dynamical System}
Observability, the feasibility of reconstructing system states from the time history of its outputs, is essential in design of control systems. Observability Gramian is the most common construct for evaluating system observability. For a linear system 
\begin{equation}
   \dot{\vec{x}} = A \vec{x} + B \vec{u},  \ \ \vec{y} = C \vec{x} \,,  \label{linear}
\end{equation}
the observability Gramian \cite{Krener09},
\begin{equation}
	W_O = \int_{0}^{t_f} e^{A^Tt} C^T C e^{At} \mathrm{d}t \,, \label{LinGram}
\end{equation}
can be computed to evaluate the observability of a system. A system is fully observable if and only if the corresponding observability Gramian is full rank \cite{Muller72}.

\begin{proposition}
Minimizing the trace of the observability Gramian is equivalent to minimizing the inverse of the estimation error covariance matrix for a linear system.
\end{proposition}

\begin{proof}
See Theorem 1 of \cite{alaeddini2016observability}.
\end{proof}

An alternative method to evaluate observability of a nonlinear system is using a relatively new concept of the observability covariance, namely the \emph{empirical observability Gramian} \cite{hahn2002improved, Krener09}. This tool provides a more accurate description of a nonlinear system's observability, while it is less computationally expensive than some other approaches, e.g.\,, Lie algebraic based approaches. For a given small perturbation $\epsilon > 0$ of the state, let $\vec{x}_0^{\pm i} = \vec{x}_0 \pm \epsilon \vec{e}_i$ be the initial condition and $\vec{y}^{\pm i}(t)$ be the corresponding output, with $\vec{e}_i$ is the $i^{\text{th}}$ unit vector in $\mathbb{R}^n$. For the system
\begin{empheq}[left=\Sigma_1 :\empheqlbrace ]{align}
 &    \dot{\vec{x}}=\vec{f}(\vec{x},\vec{u}), & \vec{x} & \in \mathbb{R}^n, \ \ \vec{u} \in \mathbb{R}^p \nonumber \\
 &	 \vec{y}=\vec{h}(\vec{x}),  & \vec{y} & \in \mathbb{R}^m\,, \label{generalNL}
\end{empheq}
the empirical observability Gramian, $W_O$, is an $n \times n$ matrix, whose $(i,j)$ component, $W_{O_{ij}}$, is given by 
\begin{equation}
	\frac{1}{4\epsilon^2} \int_{0}^{\infty} \left(\vec{y}^{+i}(t)-\vec{y}^{-i}(t) \right)^T \left(\vec{y}^{+j}(t)-\vec{y}^{-j}(t) \right) \mathrm{d}t. \label{EmpObsGram}
\end{equation}
It can be shown that if the system is smooth, then the empirical observability Gramian converges to the local observability Gramian as $\epsilon \to 0$. 
The largest eigenvalue \cite{Singh05}, smallest eigenvalue \cite{Krener09}, the determinant \cite{DeVries13, Serpas13}, and the trace of the inverse \cite{DeVries13} of the observability Gramian have all been used as measures for observability.

\subsection{Online Convex Optimization}

Online learning has recently become popular for tackling very large scale estimation problems. The convergence properties of these algorithms are well understood and have been analyzed in a number of different frameworks, including game theory \cite{hazan2007logarithmic}. In the game theoretic framework, for analyzing online learning, a game called \emph{Online Convex Optimization} problem, which is a special case of a general online optimization problem, is defined. In this game, a player chooses a point from a convex set. After choosing the point, an adversary reveals a convex loss function, and the online player receives a loss corresponding to the point she had chosen. This scenario is then repeated. To measure the performance of the online player in this game, a metric called \emph{regret} is used. Regret is the difference between the loss of the online player and the best fixed point in hindsight. The algorithm introduced in \cite{hazan2007logarithmic} is similar to the well known Newton-Raphson method for offline optimization problem. This algorithm attains regret which is proportional to the logarithm of the number of iterations when the loss (cost) functions are convex.

Assume that an online player iteratively chooses a point from a non-empty, bounded, closed and convex set in the Euclidean space denoted by $\mathcal{P}$ for $T$ times. At iteration $t$, the algorithm takes the history of cost functions, $\{f_1,\cdots, f_{t-1}\}$ as input, and suggests $x_t \in \mathcal{P}$ to the online player to choose. After committing to this $x_t$, a convex cost function $f_t : \mathcal{P} \rightarrow \mathbb{R}$ is revealed. The loss associated with the choice of the online player is the value of the cost function at the point she had committed to, i.e.\,, $f_t(x_t)$. The objective is to minimize the accumulative penalty. The regret of the online player at the end of the game, i.e.\,, at time $T$, is defined as the difference between the total cost and the cost of the best single decision, where the ``best'' is chosen with the benefit of hindsight. Formally, regret is defined as
\begin{equation}
\mathcal{R}_T = \sum_{t=1}^T f_t(x_t) - \underset{x \in \mathcal{P}}{\text{min}} \sum_{t=1}^T f_t(x)\,.
\end{equation}
%
The objective of the online algorithm is to achieve a guaranteed low regret. Specifically, it is desired that the online algorithm guarantees a sub-linear $\mathcal{R}_T$ or $\displaystyle \frac{\mathcal{R}_T}{T} \rightarrow 0$. A sub-linear regret guarantees that on average the algorithm performs as well as the best fixed action in hindsight.

\section{System Modeling}

Consider a multi-agent network consisting of $N$ agents and let $\mathcal{G}$ denote the inter-agent communication graph, where each agent $i$ has the ability to transmit scalar data to its neighbors, denoted by $\mathcal{N}(i) = \{j : (i, j) \in \mathcal{E}\}$. The state dynamics of each agent in the network is assumed to be stable and linear, and the neighboring agents $i$ and $j$ are connected by a linear diffusion with weight $w_{ij}$.

The network assumed in this paper is a network of multiple identical agents. The scalar state of each agent follows a stable dynamics, $\dot{x}_i = - x_i$ augmented with a weighted consensus dynamics. Specifically, the dynamics of the network is given by:
\begin{equation} \label{networkLTI}
\begin{aligned}
	\dot{\vec{x}} &= \left\{ \sum_{l=1}^M \left[ -\left( \vec{e}_{ij}\right) \left( \vec{e}_{ij}\right)^T w_{ij} \right] - I \right\} \vec{x}\,, \ \	l=\sigma(i,j)\,. 
\end{aligned}
\end{equation}
The dynamics of the network can then be written as 
\begin{equation}
	\dot{\vec{x}} = A(\vec{w}) \vec{x} = \left( A_0+\sum_{l=1}^M A_l w_{ij} \right) \vec{x} \,, \label{wParameter}
\end{equation}
where,
\begin{equation} \label{A0Als}
	A_0 = - I, \ \  A_l = -\vec{e}_{ij} \vec{e}_{ij}^T\,, \ \ l = \sigma(i,j)\,.
\end{equation}

\begin{lemma} \label{A_nd}
The matrix $A(\vec{w})$ is negative definite for all positive weights.
\end{lemma}

\subsection{Optimization Problem}

Now, assume an intruder in the network aims to retrieve information by connecting to an agent in the network, say agent $k$. Then the data being directly exposed to the intruder is $y = x_k$. The objective here is to minimize the information leaking from this node through an adaptation mechanism on the network. Thus, we proceed to minimize a particular metric of observability of the network with respect to the measurement $y = x_k$. This optimization problem for privacy can be written as
\begin{equation} \label{MaxPrivacy}
\begin{aligned}
& \underset{\vec{w}}{\text{minimize}}
& & \trace \left(W_{O}\right) \\
& \text{subject to}
& & \dot{\vec{x}} = A(\vec{w}) \vec{x} \,\\
& \text{}
& & y = x_k \,.
\end{aligned} 
\end{equation}

Now, let us discretize the time period $[0,t_f]$, such that $0=t_0<t_1<\cdots<t_f$. Thus, the trace of the empirical observability gramian is given by
\begin{equation}
\begin{aligned}
\trace\left(W_{O}\right) &= \frac{1}{4\epsilon^2} \int_{0}^{t_f} \sum\limits_{i=1}^N \left\| y^{+i}(\tau)-y^{-i}(\tau) \right\|^2 \mathrm{d}\tau \\
&=\sum_{t}  \int_{t}^{t+\Delta} \frac{1}{4\epsilon^2} \sum\limits_{i=1}^N \left\| x_k^{+i}(\tau)-x_k^{-i}(\tau) \right\|^2 \mathrm{d}\tau \\
&=\sum_{t}  \int_{t}^{t+\Delta} \sum\limits_{i=1}^N \left( \vec{e}_i^T e^{A\tau} \vec{e}_k \right)^2 \mathrm{d}\tau\\
&=\sum_{t}  \int_{t}^{t+\Delta} \trace \left( \vec{e}_k \vec{e}_k^T e^{2A\tau}  \right) \mathrm{d}\tau\\
&=\sum_{t}  \int_{t}^{t+\Delta} \left[e^{2A\tau}\right]_{k,k} \mathrm{d}\tau \,.
\end{aligned} \label{TrW_LTV}
\end{equation}

\begin{lemma} \label{partialCVX}
The function $\phi(\vec{w},\tau)=\left[e^{2A(\vec{w})\tau}\right]_{k,k}$ is convex with respect to $\vec{w}>0$, for all $\tau>0$.
\end{lemma}
\begin{proof}
In the Appendix, \cref{CVX_expMatrix}, it is proved that $\left[e^{A(\vec{w})}\right]_{k,k}$ is convex for all matrices $A(\vec{w})=A_0(\vec{w})+\sum_{l} A_l(\vec{w}) \vec{w}[l]$, for $A(\vec{w})$ defined in \eqref{wParameter}.
\end{proof}

%
\subsection{Online Algorithm for Weight Selection}

We now consider the problem of online adaptation of the weight in the network to minimize
\begin{equation}
\begin{aligned}
& \text{minimize}
& & f_t(\vec{w}) \\
& \text{subject to}
& & \vec{w}_{min} \leq \vec{w} \leq \vec{w}_{\max} \\
& \text{}
& & \vec{1}^T \vec{w} = 1 \,,
\end{aligned} \label{OnlineMaxObs}
\end{equation}
where 
\begin{equation} \label{costOnline}
	f_t(\vec{w}) = \int_{t}^{t+\Delta} \phi(\vec{w},\tau) \mathrm{d}\tau \,.
\end{equation}

\begin{theorem}
The cost function \eqref{costOnline} is convex.
\end{theorem}
\begin{proof}
In \cref{partialCVX}, we have shown that $\phi(\vec{w},\tau)$ is convex. Since, integral of a convex function is convex, the result of the theorem follows.
\end{proof}

\begin{theorem}
The cost function \eqref{costOnline} has bounded gradients.
\end{theorem}
\begin{proof}
First let us calculate the derivative of $\phi(\vec{w},\tau)$ with respect to $w_{ij}$,
\begin{equation} \label{gradientwij}
\begin{aligned}
&\frac{\partial \phi(\vec{w},\tau)}{\partial w_{ij}} = \vec{e}_k^T \left( 2\tau e^{2A(\vec{w})\tau} A_l \right) \vec{e}_k \\
&= -\vec{e}_k^T \left( 2\tau e^{2A(\vec{w})\tau} \vec{e}_{ij} \vec{e}_{ij}^T \right) \vec{e}_k \\
&= -2\tau \trace \left[\left( \vec{e}_k \vec{e}_k^T e^{A(\vec{w})\tau} \right)^T \left(\vec{e}_{ij} \vec{e}_{ij}^T e^{A(\vec{w})\tau} \right) \right]\\
& = -2\tau \sum_{m,n} \left[ \left( \vec{e}_k \vec{e}_k^T e^{A(\vec{w})\tau} \right) \circ \left( \vec{e}_{ij} \vec{e}_{ij}^T e^{A(\vec{w})\tau} \right) \right]_{m,n} \,.
\end{aligned}
\end{equation}
Thus,
\begin{equation}
\nabla \phi(\vec{w},\tau) = -2\tau \text{ diag} \left( E^T e^{2A\tau} \vec{e}_k \vec{e}_k^T  E \right) \,,
\end{equation}
and,
\begin{equation}
\vec{g}_t(\vec{w})=\nabla f_t(\vec{w}) = \int_{t}^{t+\Delta} \nabla \phi(\vec{w},\tau) \mathrm{d}\tau \,.
\end{equation} 
Since $A$ is a negative definite matrix for all $\vec{w}$ (\cref{A_nd}), $\lambda_i\{A\} < 0$; thus the matrix $\tau e^{2A\tau}$ has finite entries for all $\tau>0$. Therefore, $\nabla \phi(\vec{w},\tau)$ is bounded, and the integral over a finite period of time, $\tau \in [t,t+\Delta]$ is bounded. Therefore, there always exists an upper-bound, $G$, for the gradient such that
\begin{equation}
\begin{aligned}
& \underset{\vec{w} \in \mathcal{P}, t \in [T]}{\text{sup}}
& & \| \nabla f_t(\vec{w}) \|_2 \leq G \,.
\end{aligned}
\end{equation}
\end{proof}

\subsection{Regret Minimization}
Regret is the difference between the cost of the sequence of actions and the performance of the best single action, $\vec{w}^*$, taken at every time step, considering $f_t(\vec{w})$ is known for all time a priori. The regret of an action sequence $\{\vec{w}_t\}$ is 
\begin{equation}
\mathcal{R}_T = \sum_{t=1}^T \left( f_t(\vec{w}_t) - f_t(\vec{w}^*) \right)\,.
\end{equation}

Here, we are interested in applying the celebrated $\log(T)$ bound online algorithm \cite{hazan2007logarithmic}. The \emph{Online Newton Step} (ONS) algorithm is given in Algorithm \ref{ONS_Alg}. This algorithm is straightforward to implement, and the running time is $O(M)$ per iteration given the gradient. This algorithm is based on the well known Newton-Raphson method for offline optimization uses second-order information of the cost functions. 
To implement this algorithm, the convex set $\mathcal{P}$ should be bounded, such that 
\begin{equation}
D = \underset{\vec{x}, \vec{y} \in \mathcal{P},}{\text{max}} \| \vec{x}-\vec{y} \|_2 < \infty \,;
\end{equation}
$D$ is called the diameter of the underlying convex set $\mathcal{P}$. 
\begin{algorithm}[!h]
Inputs: convex set $\mathcal{P} \subset \mathbb{R}^M$, initial $\vec{w}_1 \in \mathcal{P}$\;
Set $\displaystyle \beta := \frac{1}{8GD}$ and $\displaystyle \epsilon = \frac{1}{\beta^2 D^2}$\;
In iteration $s=1$: use point $\vec{w}_1 \in \mathcal{P}$\;
 \Repeat{Convergence}{
 	$\displaystyle A_s = A_0+\sum_{l=1}^M A_l \vec{w}_s[l]$\;
	$\displaystyle \nabla \phi(\vec{w}_s,\tau) = -2\tau \text{ diag} \left( E^T e^{2A\tau} \vec{e}_k \vec{e}_k^T  E \right)$\;
		$\displaystyle \vec{g}_{s} := \int_{s}^{s+\Delta} \nabla \phi \mathrm{d}\tau $\; 
		$\displaystyle \mathbb{A}_s : = \sum_{i=1}^s \vec{g}_{i} \vec{g}_{i}^T + \epsilon I_M$\;
		$\displaystyle \vec{w}_{s+1} = \Pi_{\mathcal{P}} \left( \vec{w}_s - \frac{1}{\beta} \mathbb{A}_s^{-1} \vec{g}_{s} \right)$\;
		$s := s+1$\;
		Here, $\displaystyle \Pi_{\mathcal{P}}$ denotes the \emph{projection} in the norm induced by $\displaystyle \mathbb{A}_s$, $\Pi_{\mathcal{P}}^{\mathbb{A}_s}(\vec{y}) = arg\min_{\vec{x} \in \mathcal{P}} (\vec{y}-\vec{x})^T \mathbb{A}_s (\vec{y}-\vec{x})$\;
\vspace{3mm}
 }
 \vspace{3mm}
 \caption{Online Newton Step Algorithm} \label{ONS_Alg}
\end{algorithm}

Using the Online Newton Step algorithm which is a second-order method is not essential here. We can easily apply a first-order online learning algorithm, e.g.\,, Online Gradient Descent \cite{hazan2007logarithmic}, and attain regret proportional to the logarithm of the number of iterations. A benefit of using Online Newton Step algorithm is its faster convergence compared to the first order algorithms. The main advantage of a first order algorithm compared to second order algorithms is its ease of generalization to the distributed scenarios. 

\section{Simulation}

Our privacy guarantee online algorithm is tested on a random 9 node graph. We assume there exists a single foreign agent attacking the network. The location of the foreign agent is fixed from $t = 0$ to $t = 25$, where the intruder changes its location to another node. The best graph topology for this foreign agent evolution is calculated over 50 iterations. The resultant regret is depicted in \cref{25regret} emphasizing the performance agreement with the $O(\log(T))$ bound found by Hazan \emph{et al.} \cite{hazan2007logarithmic}.
\begin{figure}
   \begin{center}
  	 {\includegraphics[width=.5\textwidth]{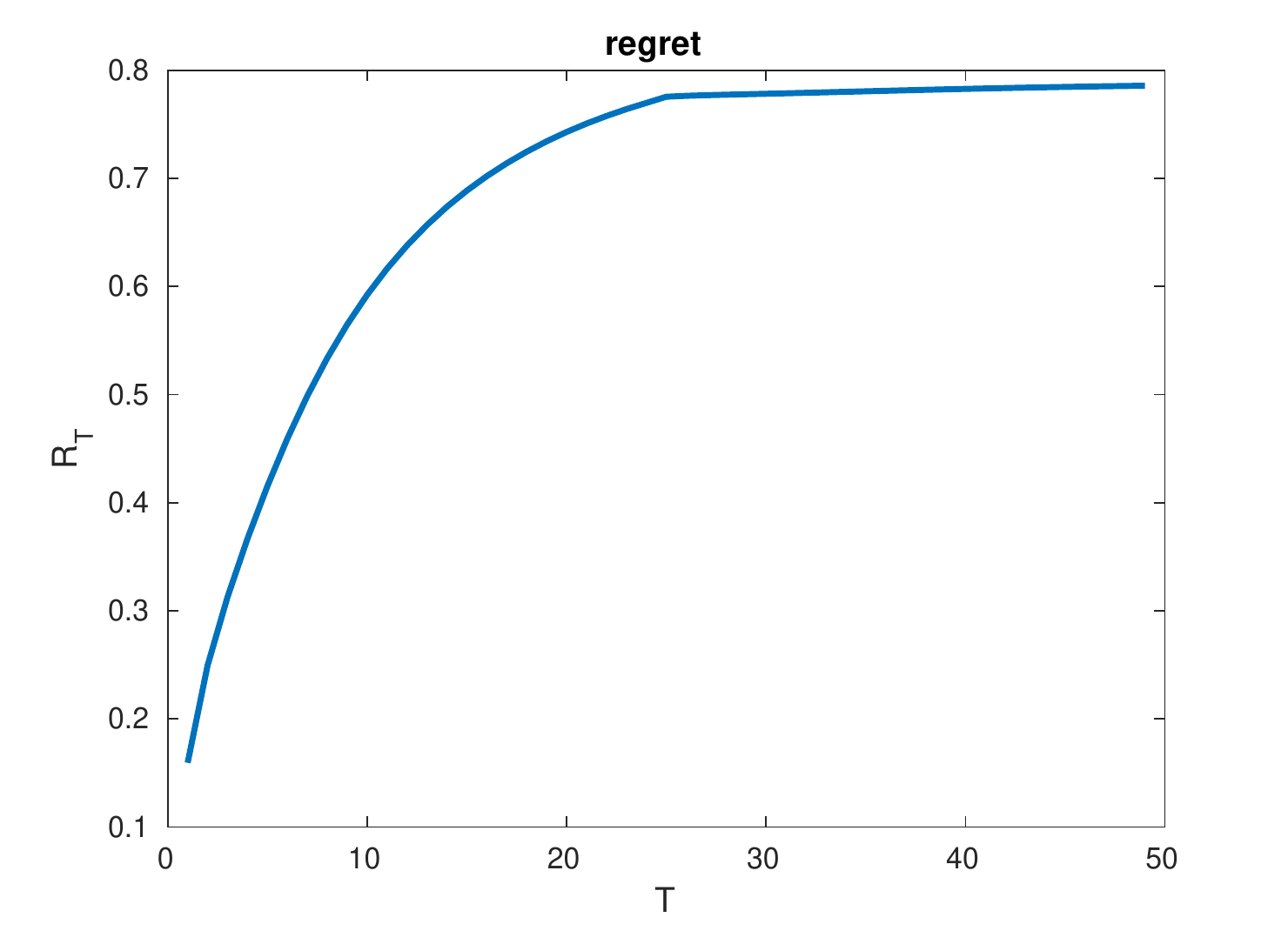}}
   \end{center}
   \caption{Regret over time of Algorithm \ref{ONS_Alg}.}
   \label{25regret}
\end{figure}

The online Newton step algorithm, given in Algorithm \ref{ONS_Alg} is also applied to a 15 node graph with multiple foreign nodes that their locations are randomly changing over time.
The location of foreign agents are detected by the agents marked with large circles. The edge weights are initialized randomly and constrained such that $\vec{1}^T \vec{w}=1$, $\vec{w}_{min} = 0.01$ and $\vec{w}_{\max} = 0.99$. The resultant graph after each online Newton step is depicted in \cref{multi10change}. The first three figures (top of \cref{multi10change}) depict the algorithm response to static foreign agents locations. The notable characteristic, over these figures, is the increasing of edge weights on those edges close to the foreign agents. The intuition behind this behavior is that the state of the node which is connected to the malicious node need to synchronize fast to make it indistinguishable to the foreign agent. Note that the state of a node when it reaches to its equilibrium becomes less distinguishable, and the measure of observability decreases. At time $t=10$, the location of the red circles (two of foreign agents) change, and the three plots on the bottom depict the algorithm's response to this change in the foreign agent location. Note that it is assumed here that the networked system is stable; therefore, the states of all nodes in the network synchronizes as time passes. Thus, adding an intruder after the network becomes synchronous, does not change the weights significantly. This comes from the fact that the state of a system becomes indistinguishable when it approaches to its equilibrium point.
\begin{figure*}
   \begin{center}
  	 {\includegraphics[width=1.0\textwidth]{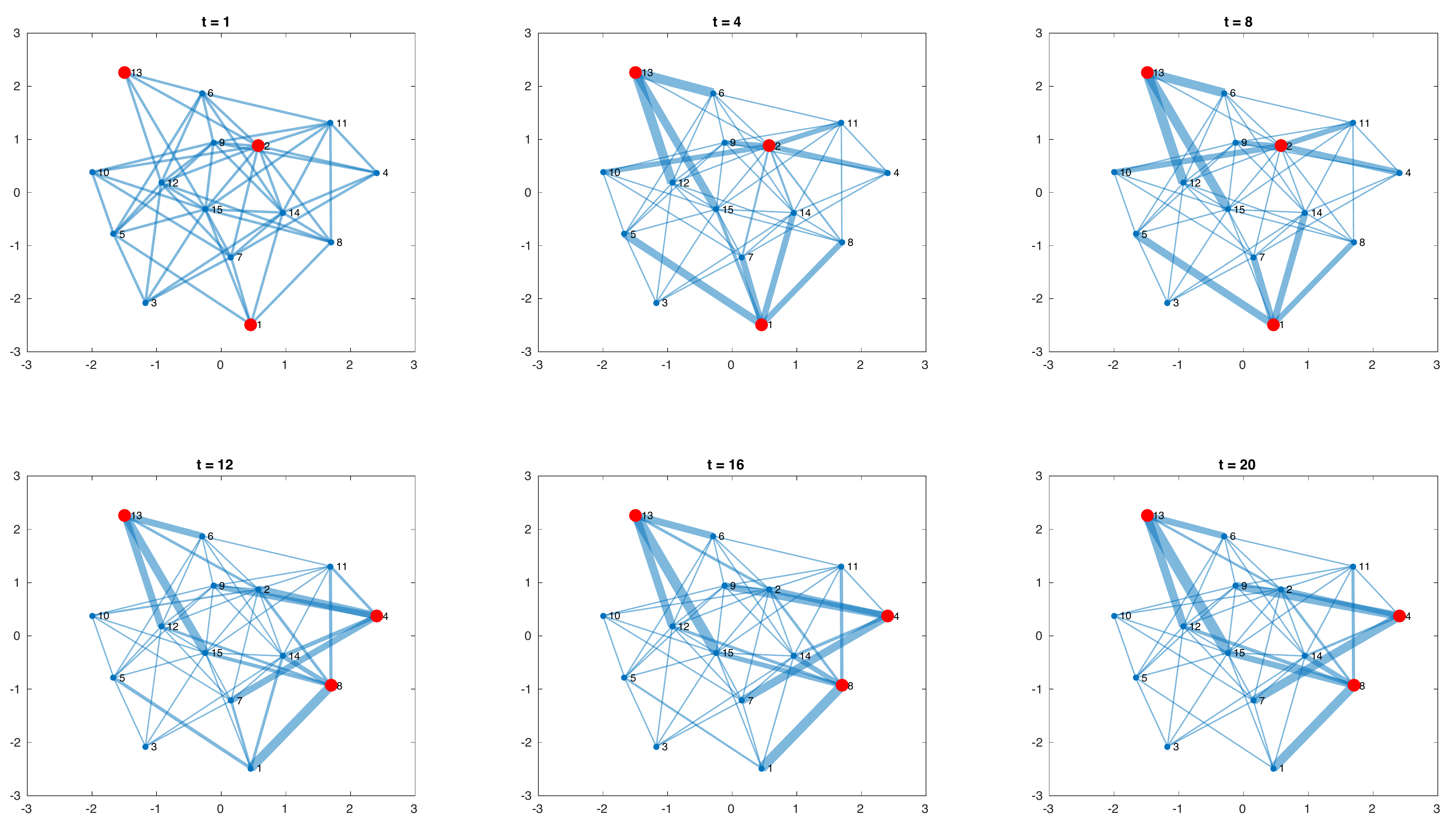}}
   \end{center}
   \caption{Re-weighting of a network with multiple foreign nodes, denoted with large red circles.}
   \label{multi10change}
\end{figure*}


\section{CONCLUSIONS}
\label{sec:conclusion}

The main goal in this paper was to design a proper selection of the weights in the dynamics of a networked system induced by the communication graph, such that the privacy of the network is guaranteed, where the privacy guarantee makes each agent cannot retrieve the initial states of other agents. It was assumed that we are uncertain about the intention of the foreign agent who attacks the network. The privacy of the network was posed as an online optimization problem on the gramian of a weighted network. An online learning algorithm was used to find an optimal set of weights that guarantee sub-linear regret. The use of empirical observability Gramian was discussed in the context of privacy for a networked control systems. For this approach it may be worth noticing that the empirical observability Gramian can be used to evaluate the observability of nonlinear systems, however, we considered a parametrized linear system to guarantee a logarithmic regret bound.



\bibliography{citations}
\bibliographystyle{IEEEtran}

\section{APPENDIX}

Consider the function $A(\vec{x}):\mathbb{R}^M \rightarrow \mathbb{R}^{N \times N}$ defined as $A(\vec{x}) = A_0(\vec{x})+\sum_l A_l \vec{x}[l]$, such that all matrices $A_i(\vec{x})$ are symmetric, and $A(\vec{x})$ is invertible for all $\vec{x} \in \mathbb{R}^M$. It is also assumed that for all $\vec{x}, \vec{y} \in \mathbb{R}^M$, $A(\vec{x})$ and $A(\vec{y})$ commute. 
%
Define the matrix-valued function $F:\mathbb{R}^M \rightarrow \mathbb{R}^{N \times N}$ as $F(\vec{x}) = e^{A(\vec{x})}$. Notice that $F$ is positive definite, because it is the exponential of a symmetric matrix. 

\begin{lemma} \label{Z_inequality}
$e^Z \succeq I+Z$ for all symmetric matrices $Z$.
\end{lemma}

\begin{proof}
Since $Z$ is a symmetric matrix, it can be decomposed as $Z=U\Lambda U^T$, where $U$ is an orthogonal eigenvectors matrix, and $\Lambda$ is a diagonal matrix whose entries are the eigenvalues of $Z$. Then, we have $e^Z =U e^\Lambda U^T$. 
It remains to prove that $e^\Lambda \succeq I+\Lambda$. Since, $e^\lambda \geq 1+\lambda, \forall \lambda \in \mathbb{R}$, thus, $e^\Lambda \succeq I+\Lambda$, which completes the proof of the lemma.
\end{proof}

%

\begin{theorem} \label{CVX_expMatrix}
The function $f(\vec{x}) = \vec{a}^T F(\vec{x}) \vec{a}$, for all $\vec{a} \in \mathbb{R}^N$ is convex.
\end{theorem}

\begin{proof}
Recall that the differentiable function $f$ is convex if and only if
$$f(\vec{y}) \geq f(\vec{x})+f'(\vec{x}) \cdot (\vec{y}-\vec{x}) \,,$$
for all $\vec{x}, \vec{y}$. Hence, we need to show that
\begin{equation} \label{needToProve}
e^{A(\vec{y})} - e^{A(\vec{x})} \succeq \sum_{i=1}^M \frac{\partial e^{A(\vec{x})}}{\partial x_i} (y_i - x_i) \,,
\end{equation}
for all $\vec{x}, \vec{y} \in \mathbb{R}$. The derivative of the exponential function $e^{A(\vec{x})}$ is given by
\begin{equation}
\begin{aligned} 
&\frac{\partial e^{A(\vec{x})}}{\partial x_i} = e^{A(\vec{x})} \frac{1-e^{-\text{ad}_{A(\vec{x})}}}{\text{ad}_{A(\vec{x})}} A_i \\
&= e^{A(\vec{x})} \sum_{k=0}^\infty \frac{(-1)^k}{(k+1)!} \left( \text{ad}_{A(\vec{x})} \right)^k A_i \\
&= e^{A(\vec{x})} \left( A_i + \frac{-1}{2!} [A,A_i] + \frac{1}{3!} [A,[A,A_i]] + \cdots \right) \,,
\end{aligned}
\end{equation}
then the term $\displaystyle \sum_{i=1}^M \frac{\partial e^{A(\vec{x})}}{\partial x_i} (y_i - x_i) $ can be written as:
\begin{equation}
\begin{aligned} 
& \sum_{i=1}^M \frac{\partial e^{A(\vec{x})}}{\partial x_i} (y_i - x_i) = \sum_{i=1}^M e^{A(\vec{x})} \left( A_i (y_i - x_i) \right. \\
& \left. + \frac{-1}{2!} [A,A_i] (y_i - x_i) + \frac{1}{3!} [A,[A,A_i]] (y_i - x_i) + \cdots \right)
\end{aligned}
\end{equation}
\begin{equation} \label{manipulation}
\begin{aligned} 
&= e^{A(\vec{x})} (A(\vec{y})-A(\vec{x}) ) + \frac{-1}{2!} e^{A(\vec{x})} \sum_{i=1}^M \left( A A_i y_i - A A_i x_i \right. \\
&\left. - A_i A y_i + A_i A x_i \right) + \cdots = e^{A(\vec{x})} (A(\vec{y})-A(\vec{x}) ) \\
&+ \frac{-1}{2!} e^{A(\vec{x})} \left[ A(\vec{x}) \left( A(\vec{y}) - A(\vec{x}) \right) - \left( A(\vec{y}) - A(\vec{x}) \right) A(\vec{x}) \right] \\
& + \cdots = e^{A(\vec{x})} (A(\vec{y})-A(\vec{x}) ) + 0 + 0 + \cdots \,.
\end{aligned}
\end{equation}
Now, if we substitute \eqref{manipulation} in \eqref{needToProve}, we need to prove
\begin{equation} 
e^{A(\vec{y})} - e^{A(\vec{x})} \succeq e^{A(\vec{x})} (A(\vec{y})-A(\vec{x}) ) \,.
\end{equation}
Since $e^{A(\vec{x})} $ is invertible, this inequality can be rearranged to give
\begin{equation} 
e^{A(\vec{y}) - A(\vec{x})} \succeq I + A(\vec{y})-A(\vec{x}) \,.
\end{equation}
Since all matrices $A_i$ are symmetric, $A(\vec{y})-A(\vec{x})$ is symmetric for all $\vec{x}, \vec{y}$. Thus, by using \cref{Z_inequality}, the proof of this theorem follows.

\end{proof}

\end{document}